\documentclass[sigconf,nonacm]{acmart}
\setkeys{acmart.cls}{balance=false}
\usepackage{tikz}
\usepackage{array}
\usetikzlibrary{arrows.meta,positioning,fit,backgrounds}
\newtheorem{proposition}{Proposition}

\graphicspath{{figures/}}

\title{When Safety Checks Become Privacy Risks: A Privacy-Preserving Paradigm for LLM Agent Safety Auditing}
\author{Yixiang Luo}
\authornote{Yixiang Luo, Ziheng Ma, and Ke Chang contributed equally to this work and share first authorship.}
\affiliation{%
  \institution{Beijing Jiaotong University}
  \city{Beijing}
  \country{China}
}
\email{luoyixiang@bjtu.edu.cn}

\author{Ziheng Ma}
\authornotemark[1]
\affiliation{%
  \institution{Beijing Jiaotong University}
  \city{Beijing}
  \country{China}
}
\email{mzh20010113@bjtu.edu.cn}

\author{Ke Chang}
\authornotemark[1]
\affiliation{%
  \institution{Beijing Jiaotong University}
  \city{Beijing}
  \country{China}
}
\email{changke@bjtu.edu.cn}

\author{Jian Wang}
\affiliation{%
  \institution{Beijing Jiaotong University}
  \city{Beijing}
  \country{China}
}
\email{wangjian@bjtu.edu.cn}

\author{Jiqiang Liu}
\authornote{Corresponding author.}
\affiliation{%
  \institution{Beijing Jiaotong University}
  \city{Beijing}
  \country{China}
}
\email{jqliu@bjtu.edu.cn}

\renewcommand{\shortauthors}{Luo et al.}
\setcopyright{none}
\hypersetup{
  pdftitle={When Safety Checks Become Privacy Risks: A Privacy-Preserving Paradigm for LLM Agent Safety Auditing},
  pdfauthor={Yixiang Luo, Ziheng Ma, Ke Chang, Jian Wang, Jiqiang Liu},
  pdfsubject={Privacy-preserving safety auditing of LLM agents},
  pdfkeywords={LLM agents, safety auditing, input confidentiality, homomorphic encryption}
}

\begin{document}
\raggedbottom
\setlength{\emergencystretch}{1em}
\begin{abstract}
Large language model agents interact with online services through tool calls that may access private data or perform sensitive actions. 
Specialized providers can supply and maintain pre-execution safety auditing capabilities for agent operators.
However, sending task descriptions, call arguments, and interaction histories to an external auditor may expose sensitive user information.
In this work, we propose a paradigm for privacy-preserving safety auditing that separates delegated audit computation from plaintext access and execution authority.
Grounded in a three-domain auditing model, the paradigm connects trusted representation construction, encrypted audit computation, and trusted decision and execution through explicit interfaces.
These interfaces accommodate alternative representations and scorers, while encrypted computation protects audit inputs from a semi-honest auditor.
We instantiate this paradigm in SealCall, a framework combining learned audit representations, CKKS-based encrypted similarity computation, and execution control bound to audited tool calls.
Evaluated on 1,252 public binary candidates, SealCall achieves a mean F1 score of 0.840 under grouped nested evaluation, while preserving plaintext audit decisions in all measured cases. 
A representative configuration completes each encrypted scoring operation in 0.176 s.

\end{abstract}
\maketitle
\hypersetup{
  pdfauthor={Yixiang Luo, Ziheng Ma, Ke Chang, Jian Wang, Jiqiang Liu},
  pdfsubject={Privacy-preserving safety auditing of LLM agents}
}

\section{Introduction}
\label{sec:introduction}
Large language model (LLM) agents support information retrieval, personal
assistance, and business workflow automation by calling tools for email,
calendars, files, browsing, and online services~\cite{toolformer,agentdojo}.
They combine tools to retrieve information, send content, or change
business records on a user's behalf. Planning errors or prompt injection
in external content can lead them to act beyond the user's
authorization~\cite{agentharm,injecagent}. Pre-execution safety auditing
therefore examines proposed tool calls in their task context before they
produce effects, providing a basis for allowing or blocking
execution~\cite{toolsafe}.

Agent vendors such as Anthropic provide integrated
safeguards~\cite{anthropic-sandboxing}, while self-hosted agent frameworks
such as OpenClaw require operators to configure and maintain
deployment-specific protections~\cite{openclaw-security}.

Traditional computer security already relies on specialized third-party
services for threat detection and incident response~\cite{nist-incident-response}.
These providers supply professional expertise, maintain detection
capabilities, and draw on threat knowledge across customers. This
established division of labor also motivates third-party auditing for
agents: as tools, attack patterns, and deployment requirements evolve,
providers can curate audit references, update models, and operate the
scoring service. Users thereby gain access to continuously maintained
safety capabilities without building a dedicated security operation.

Applying this service model to contextual agent auditing creates a
privacy challenge. Assessing whether a call exceeds the user's
authorization often requires examining the user task, the final arguments of a proposed call,
and the prior interaction history together~\cite{toolsafe,stepguard}.
Consider an office agent allowed to read internal pricing materials but
instructed to send only the public version of a quotation to a specified
customer. Assessing whether a proposed send exceeds that authorization
requires the user's instructions, the recipient, the content, and its
source. Sending these materials to an external audit service in plaintext
can expose internal costs, customer identities, and business
communications, even if the proposed call is ultimately blocked.
We call the additional visibility created by this interface
\emph{audit-induced privacy exposure}. We study how users can obtain
professionally maintained safety auditing while keeping these inputs
confidential from the service provider.

We establish a three-domain auditing model: the \emph{agent
execution domain} runs the agent and supplies the final proposed call
and prior context; the \emph{user control domain} contains a trusted
controller that processes audit inputs, holds the secret key, and controls
tool release; and the \emph{audit service domain} receives encrypted inputs,
performs the agreed computation, and returns encrypted results. These
domains follow actual access and control permissions and can be operated
by different organizations or by services with access isolation within
one organization. Sharing context with an agent does not require
disclosing it to the audit service; information the auditor already
obtains through other channels remains part of its prior knowledge.
With a trusted controller and isolated keys, we target input
confidentiality at the audit interface against a semi-honest auditor
that follows the protocol but seeks to infer private information.

Realizing this allocation requires turning safety judgments grounded in
private context into computation over ciphertexts. The risk of sending
a quotation depends on relationships among the scope of authorization,
the recipient, and the source of the content; a tool name or isolated
arguments provide an incomplete basis for judgment. Context compression
and text embeddings can construct compact numerical
representations~\cite{selective-context,sentence-bert}, while homomorphic
encryption (HE) supports server computation on encrypted
inputs~\cite{cryptonets,ckks}. However, embeddings can still disclose
sensitive content from the original text~\cite{embedding-inversion},
and encryption cannot recover risk information lost in the representation.
Representation construction and encrypted computation must therefore be
designed together to support useful risk assessment with confidential inputs.

We propose a paradigm for privacy-preserving safety auditing that connects
context compression, vectorization, homomorphic computation, and trusted
execution control. Within the trusted domain, the controller retains
the user task and final candidate, selects relevant prior history with
source annotations, and constructs and encrypts a numerical representation.
The audit service computes on this input using reference data or model
parameters encrypted through trusted preparation, including those supplied
by a provider maintaining the audit capability. The controller decrypts
the result, makes a decision, and binds release to the audited call and
its expected state. This yields three processing layers:
\emph{trusted representation construction, encrypted audit computation,
and trusted decision and execution}. SealCall combines text embeddings
with structured features for tools, arguments,
and sources through a learned projection. It uses CKKS to compute
similarities between the encrypted query and encrypted unsafe reference
examples. The trusted controller takes the maximum of the decrypted
similarities, applies a calibrated threshold, and checks that execution
matches the audited operation.

We evaluate SealCall's detection quality, representation utility,
encrypted-decision agreement, cost, and execution binding. On 1,252 binary
candidate calls from public trajectories, grouped nested evaluation yields
mean F1 of 0.840. In separate encrypted measurements, decisions agree
with frozen plaintext decisions across all 760 model--layout records.
A representative configuration with 256 dimensions and 16 references
requires 0.176\,s per encrypted scoring operation on average.
These evaluations examine trade-offs among
representations, scoring functions, and computation costs within the paradigm.

We make three contributions:
\begin{itemize}
  \item \textbf{Problem formulation for confidential outsourced auditing.}
  We connect the value of specialized audit services to the exposure of
  users' tasks, calls, and histories, and formalize input confidentiality
  against a semi-honest auditor across three domains of access and control.
  \item \textbf{A paradigm for privacy-preserving safety auditing.}
  We develop three layers for trusted representation construction,
  encrypted audit computation, and trusted decision and execution.
  We explain their interface dependencies, the choice of delegated
  computation and returned information, and the conditions for input
  confidentiality.
  \item \textbf{SealCall and empirical evaluation.}
  We develop the SealCall framework and assess
  representation utility, alternative scorers, encrypted-decision
  agreement, computation and communication costs, and execution binding.
\end{itemize}

\section{Threat Model and Problem Formulation}
\label{sec:background}
\subsection{Threat model}
\label{sec:threat-model}
Our threat model includes the user, the LLM agent, the audit service
provider, and an additional observer seeking to recover audit rules
from the encrypted reference database. The audit provider is semi-honest:
it follows the agreed computation but attempts to infer private content
from its observations. It holds public and evaluation keys, enabling
it to compute on encrypted inputs and return encrypted results.
The additional adversary is a passive observer that obtains encrypted
reference data without a decryption key and seeks to infer the audit
rules encoded in those data. Reference vectors remain encrypted during
storage and computation, preserving their contents and the audit-rule
information they encode under the encryption assumptions.

\emph{Input confidentiality} is the primary objective: protect users'
tasks, call arguments, histories, and their numerical representations
while enabling external audit computation. \emph{Reference-database
confidentiality} additionally protects reference contents against observers
with access to ciphertexts but no key or prior plaintext knowledge.
References remain encrypted during storage, transmission, and homomorphic
computation. A provider supplying its own references already knows their
contents; this prior knowledge leaves user input confidentiality as the
primary protection in that deployment. Protection of references against
the decrypting controller depends on the output interface
(Section~\ref{sec:delegation}).

Let $S$ denote the audit server. Its audit-message view is
\begin{equation}
 \mathcal V_S=\bigl(pk,evk,\operatorname{Enc}(\Theta),
                  \{\operatorname{Enc}(z_t),c_t\}_t\bigr).
 \label{eq:view}
\end{equation}
Here $pk$ and $evk$ are public and evaluation keys, $z_t$ is the audit
representation, $\Theta$ contains protected audit data or model
parameters, and $c_t$ is the encrypted response. We analyze the content
confidentiality of this message view conditioned on public metadata
and the parties' prior information, comparing different input or
reference contents under the same public layout and request schedule.

\subsection{Audit inputs and the three-domain auditing model}
\label{sec:audit-model}
A proposed tool call is an operation awaiting execution, including its
tool name and final arguments. At step $t$, the audit context is
\begin{equation}
 A_t=(u,a_t,h_{<t}),
 \label{eq:context}
\end{equation}
where $u$ is the user's task instruction and the basis for assessing
authorization, $a_t$ is the tool call awaiting execution, and $h_{<t}$
contains previous tool calls and their returned content. Source roles
distinguish user instructions from external information returned by tools.
The audit function combines these inputs to assess risk, and the controller
uses its result to allow or block the call.

Our three-domain auditing model separates responsibilities into three
domains. The \emph{agent
execution domain} runs the agent and
supplies its final candidate and prior context before tool execution.
The \emph{user control domain} contains trusted controller $C$, which
reads $A_t$, constructs and encrypts representations, holds the decryption
key, and decides whether to release a call. It can be local or run in a
trusted hosted environment. The \emph{audit service domain} contains server
$S$, which receives encrypted representations, performs the agreed
numerical computation, and returns encrypted results. It holds neither
the decryption key nor tool release authority. Audit data or model
parameters $\Theta$ are prepared and encrypted by $C$ or a cooperating
data owner in a trusted environment. A specialized audit provider can act
as this data owner, maintaining references or model parameters and
encrypting them under $C$'s public key for use by $S$. This preparation
uses the provider's audit knowledge without access to users' plaintext
queries. SealCall uses a reference set $B$ as its choice of $\Theta$.
These domains define the boundary for protocol design and confidentiality
analysis.

A common runtime interface connects the paradigm to agent systems: obtain
the final call and its context before side effects, and route covered
tools through the controller's executor. Context compression, text
embedding, and other language processing run in the trusted domain.
When one cloud provider operates both services, trust boundaries follow
actual information access and key isolation. Protection addresses the
additional exposure introduced by the audit interface.

\subsection{Requirements for a useful confidential audit}
The paradigm has four requirements. First, numerical representations
should retain the context needed to distinguish authorized and unsafe
actions, including relationships between task constraints and call
arguments. Second, the service must perform its agreed computation without
plaintext representations or a decryption key. Third, for the selected
scorer, threshold, and computation parameters, approximation error should
be small enough to preserve the controller's plaintext decisions. Fourth,
release must bind the decision to the audited call and expected state;
changes require another audit. These requirements are assessed through
detection experiments, confidentiality analysis, numerical error and
decision-agreement measurements, and execution-binding checks.

Section~\ref{sec:method} develops the computational allocation for these
requirements and presents the SealCall framework.

\section{A Paradigm for Confidential Safety Auditing}
\label{sec:method}
\subsection{Design principle and computational decomposition}
\label{sec:decomposition}
Our central design principle separates delegation of audit computation
from plaintext access and execution authority. An audit depends on
relationships expressed in private language: who requested an action,
what the request permits, and how the proposed call uses earlier
information. The trusted controller encodes the task, candidate call,
and source-linked history as a numerical input. The service evaluates
this input against audit data or model parameters, which can be supplied
and maintained by a specialized provider. Computation proceeds over
ciphertexts; the controller interprets the result and enforces the
decision. This decomposition creates three processing layers:
\emph{trusted representation construction, encrypted audit computation,
and trusted decision and execution}. The first and third run in the user
control domain; the middle layer runs in the audit service domain.
Figure~\ref{fig:paradigm} shows the privacy exposure motivating this design
and how SealCall implements the three layers.
\begin{figure*}[t]
\centering
\includegraphics[width=\textwidth]{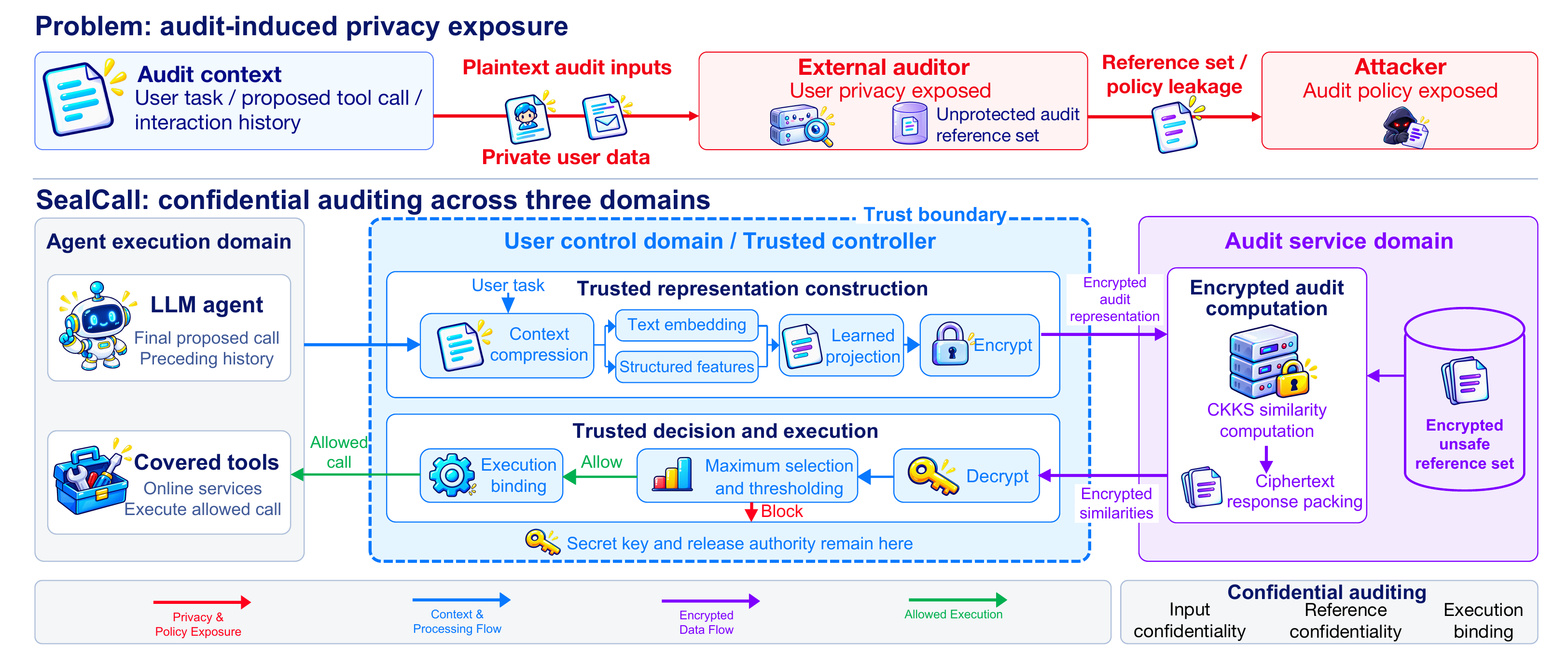}
\caption{Privacy exposure in outsourced auditing and the SealCall workflow.
Top: plaintext audit inputs expose private user context to the auditor,
while an unprotected reference set exposes audit information.
Bottom: the SealCall framework separates trusted representation construction,
encrypted audit computation, and trusted decision and execution across three
domains. The illustrated configuration uses CKKS reference similarities;
the secret key and release authority remain in the user control domain.}
\Description{The upper panel shows plaintext audit context flowing to an
external auditor and an unprotected reference set exposed to an attacker.
The lower panel shows the agent execution domain on the left, the user control
domain in the center, and the audit service domain on the right. A dashed
boundary encloses the trusted controller. The agent supplies its final proposed
call and preceding history; the user supplies the task. The controller
compresses this context, combines text embeddings and structured features
through a learned projection, and encrypts the resulting audit representation.
The audit service computes CKKS similarities using an encrypted unsafe reference
set and packs the ciphertext results for return. The controller decrypts them
and applies maximum selection and thresholding. An allow proceeds through
execution binding to the covered tools; a block stops the call.
The secret key and release authority stay within the trusted controller.}
\label{fig:paradigm}
\end{figure*}

The interfaces express how these layers compose:
\begin{align}
 q_t&=\kappa(A_t),\qquad z_t=\phi(q_t), \label{eq:interfaces}\\
 c_t&=\operatorname{Eval}_{F}(\operatorname{Enc}(z_t),
                              \operatorname{Enc}(\Theta)), \label{eq:interfacehe}\\
 b_t&=\delta(\operatorname{Dec}(c_t);\tau). \label{eq:interfacedecision}
\end{align}
Context compression $\kappa$ produces $q_t$, and vectorization $\phi$
produces the audit representation $z_t$. Audit data or model parameters
$\Theta$ are prepared in a trusted environment and encrypted under the
controller's key. The public arithmetic function $F$ specifies the
service's computation; $\delta$ interprets its output using calibration
parameters $\tau$ to produce decision $b_t$. For example, $\Theta$ can be
a reference set or classifier coefficients. The controller releases
$a_t$ when the decision permits it and the execution checks pass.

The interfaces define a family of audit implementations. A deployment
selects $\kappa$, $\phi$, $\Theta$, $F$, and $\delta$ together, so that
the retained context supports the scoring task and the encrypted
arithmetic supports its decision rule. Representation quality determines
which risks can be distinguished; encryption protects the inputs supplied
to the service; execution control applies the resulting decision to the
intended operation. The following layers make these dependencies explicit.

\subsection{Trusted context processing and audit representation}
\label{sec:representation}
Context compression centers the audit on the proposed call. It selects
task constraints and prior evidence relevant to authorization and risk,
reducing irrelevant material and the cost of subsequent encoding while
retaining the basis for judgment. The user task and final candidate
remain present; selected history carries its source. In the quotation
example, the authorized recipient, proposed destination, and origin of
the message content jointly inform the decision. The same address in a
tool response and in an explicit user authorization can have different
significance. Compression must therefore retain the task constraint and
its relationship to the call along with the selected text.

This selection also controls representation cost. Long histories contain
content unrelated to the current action, while a short excerpt can omit
a restriction or prerequisite. The selection rule determines this
trade-off. Explicit source roles, temporal order, and missing-evidence
indicators help subsequent processing distinguish available evidence from
assumptions. A compression method is assessed by the risk distinctions
that survive it and by its reduction in processing cost.

Vectorization encodes the task constraints, candidate, and historical
evidence retained in the compact context as an audit representation.
Its design aims to preserve information needed to judge authorization
and risk while making numerical scoring possible without server access
to plaintext. Embeddings can represent language content; extensions of
this interface could use image or audio embeddings. Structured features
can expose tool identity, argument structure, and source relationships,
and a learned mapping can adapt their combination to the audit task.
Output dimension, scale, and normalization determine the input expected
by $F$ and the cost of encrypted evaluation. These choices are made
jointly with the scorer: similarity calculation and linear classification
may require different representations.

Both compression and vectorization operate within trusted processing.
Their output remains sensitive and is encrypted before delegation.
Changing an encoder also requires updating the corresponding reference
vectors or classifier parameters and calibration to match its meaning
and scale. Section~\ref{sec:representation-evaluation} evaluates the
resulting representation through ablations, task restatements, and
changes to authorization.

\subsection{Delegating audit computation through HE}
\label{sec:delegation}
Delegation begins by expressing the numerical audit operation as an
arithmetic function $F$ and trusted postprocessing $\delta$. This split
determines what the server computes and what its encrypted output means.
For a reference scorer, $F$ can return dot products for trusted
aggregation; for a linear scorer, it can return a linear score for a
trusted threshold comparison. Nonlinear operations assigned to the
controller can reduce encrypted-circuit cost. Choosing this split
requires considering the number of returned values, their numerical
precision, and the information available to the recipient.

Controller $C$ generates the keys and retains the secret key. Audit data
$\Theta$ are encoded and encrypted through trusted preparation, and each
query is encrypted with fresh randomness. Server $S$ receives public and
evaluation keys and the ciphertext operands, evaluates the agreed
circuit, and returns ciphertexts. The circuit description, dimensions,
and output layout are public. Parameters and reference data carried in
$\Theta$ remain encrypted during transmission, computation, and storage,
covering the retained ciphertexts considered in the threat model.

The arithmetic and encoding jointly determine the private computation
cost. Multiplication depth, vector dimension, and output count affect
the required HE parameters and communication. Packing can reduce response
size while retaining the same logical results. Approximate arithmetic
also requires matching numerical precision to the decision rule:
an error harmless for a well-separated score can change a decision near
its threshold. These dependencies motivate measuring cost and decision
agreement together for each representation--scorer combination.

The output interface fixes what the trusted recipient learns.
Returning individual similarities supports trusted maximum selection and
inspection of references; a controller able to query arbitrary basis
vectors can also recover reference coordinates from these similarities.
The present trust model permits this access within $C$. Protecting a
provider's reference assets against the controller is a separate goal:
it calls for an output-restricted protocol that reveals prescribed
predictions or search results~\cite{gazelle,sanns}. This provides a route
to extending input confidentiality to protection of both parties' data.

\subsection{Trusted decision and execution control}
\label{sec:trusted-decision}
The controller turns numerical results into an operational decision.
It decrypts the returned values, applies the specified aggregation or
postprocessing, and evaluates $\delta$ at the calibrated operating
point. Calibration relates the scorer to the deployment's preference
between detecting unsafe calls and blocking legitimate ones.
The representation, audit data, arithmetic, and decision rule form one
versioned configuration. Provider updates to references or model parameters
must match the deployed encoder and calibration; changes require renewed
validation before the controller activates the new configuration.

A decision becomes an execution constraint by being attached to a
specific pending operation. Before delegation, the controller records
the final call and expected state. At release it checks that the result
belongs to this request, that the call and relevant state still match,
and that the decision permits execution. A change triggers another audit;
an unavailable or unusable result keeps the call blocked. Routing each
covered tool through this check gives the decision a concrete enforcement
point, following complete mediation~\cite{saltzer-schroeder}.

This layer preserves the relationship between the evidence examined and
the action performed. Its validation checks stale state, replay,
argument changes, and failure handling. Detection experiments separately
measure whether the selected representation and scorer produce useful
decisions. Together these checks connect an encrypted numerical service
to control over the agent's effects.

\subsection{Confidentiality and conditions of applicability}
\label{sec:confidentiality}
Input confidentiality follows from keeping plaintext preparation and
decryption within the user control domain. Fix the public circuit,
parameters, layouts, and request schedule when comparing the message
views in Equation~\ref{eq:view}. Assume semantic security of the chosen HE scheme
in its evaluation-key setting, fresh encryption randomness, and correct
key management. Conditioned on the same prior information, two equal-shape
collections of audit representations and protected parameters yield
computationally indistinguishable ciphertext views to $S$. When the
provider supplies $\Theta$, this comparison fixes its known parameters
and varies the user's representations. A hybrid replaces the protected
input ciphertexts by encryptions of zero of the same shape. Evaluation and packing are
computations on that view and preserve indistinguishability.

The task, call arguments, and history remain in trusted processing;
their representations reach the service only as ciphertexts, including
in retained logs and backups. Public layouts and request timing remain
observable. The confidentiality claim concerns message contents under a
fixed public layout and request schedule, conditioned on prior knowledge.
Subsequent tool activity forms a separate feedback channel that can
reveal aspects of the workflow.

Applicability requires trusted context preparation, isolated keys, correct
mediation, and the semi-honest service behavior specified in
Section~\ref{sec:threat-model}. The audit server and passive observer
hold no decryption key and do not collude with its holder.
Execution binding enforces the association between a result and its call.
Correct evaluation relies on the semi-honest assumption: a malicious
server could return a well-formed ciphertext encoding a false score.
Extending the protocol to such servers requires a computation-integrity
mechanism, such as verifiable evaluation, alongside call binding.

\subsection{The SealCall Framework}
\label{sec:sealcall}
SealCall combines learned audit representations, encrypted reference scoring,
and a call-binding controller (Table~\ref{tab:instantiation}). We evaluate
its detection utility, encrypted cost, and execution consistency.
Its reference set is $\Theta=B=\{p_j\}_{j=1}^K$; each $p_j$ comes from an
actual unsafe fitting example, allowing the trusted party to inspect
the highest-scoring example.

\begin{table*}[t]
\caption{Three processing layers of the SealCall framework.
The trusted controller $C$ retains plaintext access, the secret key,
and release authority; the audit server $S$ computes on ciphertexts.}
\label{tab:instantiation}
\centering
\begingroup
\setlength{\tabcolsep}{5pt}
\renewcommand{\arraystretch}{1.12}
\begin{tabular}{@{}>{\raggedright\arraybackslash}p{0.23\textwidth}>{\raggedright\arraybackslash}p{0.32\textwidth}>{\raggedright\arraybackslash}p{0.40\textwidth}@{}}
\toprule
Layer (role) & Interface & Components in the evaluated configuration\\
\midrule
Trusted representation ($C$) &
Audit context $A_t$ $\rightarrow$ encrypted representation $\operatorname{Enc}(z_t)$ &
Source-linked context; BGE and structured features; learned projection.\\
\addlinespace
Encrypted computation ($S$) &
Encrypted representation and audit data $\rightarrow$ encrypted results &
CKKS reference dot products; ciphertext response packing.\\
\addlinespace
Trusted decision and execution ($C$) &
Encrypted results $\rightarrow$ call-bound allow or block &
Decryption; maximum and threshold; call and state checks.\\
\bottomrule
\end{tabular}
\endgroup
\end{table*}

\paragraph{Trusted representation.}
SealCall retains the user task and canonical JSON of the final candidate,
then selects preceding history by scalar argument occurrences with
source roles and offsets. It combines frozen
\texttt{bge-small-en-v1.5} embeddings~\cite{bge-cpack} with structured
features for tools, argument paths and types, history presence, and
argument provenance. Vocabularies and scaling statistics use fitting
data. With standardized input $x_t$ and $v$ varying fitting coordinates,
the learned affine projection produces a unit vector:
\begin{equation}
 \widetilde{x}_t=\sqrt{v}\,\frac{x_t}{\lVert x_t\rVert_2},\qquad
 z_t=\frac{W\widetilde{x}_t+b}{\lVert W\widetilde{x}_t+b\rVert_2}.
 \label{eq:projection}
\end{equation}
Supervised fitting trains the maximum similarity to unsafe references,
with same-root query--reference pairs masked. Nested selection chooses
the map and reference count; separate calibration sets the threshold.
Appendix~\ref{app:training} gives extraction, encoding, and fitting details.

\paragraph{Encrypted computation and trusted decision.}
CKKS~\cite{ckks} supports approximate arithmetic on encoded values with
rescaling to manage scale and precision. SealCall uses it for dot products:
\begin{equation}
 c_{t,j}=\operatorname{EvalDot}
  (\operatorname{Enc}(z_t),\operatorname{Enc}(p_j)),\quad j=1,\ldots,K.
 \label{eq:dot}
\end{equation}
TenSEAL~\cite{tenseal} packs these scalar ciphertexts into a response.
The controller decrypts all $K$ similarities and applies
\begin{equation}
 s_t=\max_{1\leq j\leq K}z_t^\top p_j,\qquad
 \widehat y_t=\mathbb{1}[s_t\geq\tau],
 \label{eq:score}
\end{equation}
where $\widehat y_t=1$ means unsafe. Normalization, maximum selection,
and thresholding run at $C$, keeping the encrypted circuit focused on
similarity arithmetic. The score measures similarity; the selected
reference provides an inspectable comparison.

\begin{proposition}[Decision preservation]
Let $v_j=z_t^\top p_j$ and let $\widetilde v_j$ be its decrypted
approximation. If $|\widetilde v_j-v_j|\leq\epsilon$ for every reference,
then $|\max_j\widetilde v_j-s_t|\leq\epsilon$. If also
$|s_t-\tau|>\epsilon$, the approximate and plaintext decisions agree.
\end{proposition}
\begin{proof}
All approximate similarities are at most $s_t+\epsilon$. The
approximation for a true maximizer is at least $s_t-\epsilon$.
The approximate maximum lies between these bounds and cannot cross a
threshold farther than $\epsilon$ away.
\end{proof}
This condition motivates the score-error and decision-agreement
measurements in Section~\ref{sec:encrypted-evaluation}.

\paragraph{Execution binding.}
Before requesting the audit, the adapter saves canonical arguments as
the pending execution snapshot. It hashes the tool name, call identifier,
per-gate nonce, and expected state version to identify the request.
Only a timely allow that matches this call and expected state releases
the saved snapshot. Replayed, duplicate, malformed, mismatched, or
timed-out results fail closed. Calls in each tool lane are serialized;
state commits only after successful execution, and uncertain side effects
halt the lane to avoid further actions on unknown state. The controller
associates state versions with the intended trusted history prefix and
keeps the audit configuration fixed for the session. Thus the digest
checks the call and expected state, while these controller checks
maintain consistency between audit evidence and execution.

\section{Evaluation}
\label{sec:evaluation}
SealCall tests whether the paradigm's interfaces form a useful audit:
whether the numerical representation preserves risk distinctions,
encrypted scoring preserves plaintext decisions at measurable computation
and communication cost, and trusted control releases the audited call.

Ablations and scorer comparisons assess information retention and
component choice (Sections~\ref{sec:representation}--\ref{sec:delegation});
encrypted measurements assess decision agreement and cost, while adapter
tests assess execution binding (Section~\ref{sec:trusted-decision}).
Section~\ref{sec:confidentiality} analyzes input confidentiality.
Candidates come from existing trajectories, so detection metrics concern
fixed calls; task completion after rejection and live attack success
require interactive evaluation.

\subsection{Data and evaluation protocol}
The data are public step annotations released by ToolSafe~\cite{toolsafe},
derived from AgentHarm and AgentDojo trajectories~\cite{agentharm,agentdojo}.
Of 1,951 source records, format validation excludes 29 and exact-input
deduplication merges 397, leaving 1,525 accepted candidates. We retain
the native labels: 1,252 binary candidates enter training and evaluation;
273 intermediate-label candidates remain separate
(Table~\ref{tab:dataset}). Normalized task bodies, source-local trajectories,
and identical visible inputs define indivisible task roots before splitting.
Each input contains its task, final candidate, and strictly prior history.

\begin{table}[t]
\caption{Accepted public candidates. Intermediate labels are excluded
from binary fitting and evaluation. Roots are the split unit.}
\label{tab:dataset}
\centering
\begin{tabular}{lrrrrr}
\toprule
Source & Binary & Roots & Safe & Unsafe & Interm.\\
\midrule
Harm & 399 & 144 & 205 & 194 & 273\\
Dojo & 853 & 58 & 512 & 341 & 0\\
Total & 1,252 & 202 & 717 & 535 & 273\\
\bottomrule
\end{tabular}
\end{table}

Five outer grouped folds evaluate a fixed fitting procedure. Each outer
training pool reserves a group-disjoint 20\% for final calibration; three
inner folds select configurations on the remaining development data.
Each fit has separate stopping groups. Fitting alone determines feature
vocabularies, scaling, learned maps, and reference examples. Final fitting,
stopping, calibration, and test partitions have no overlapping roots,
exact tasks, or exact inputs. Every binary candidate receives one outer-test
prediction per seed. These public data had been inspected during method
development, so the protocol measures performance on an exposed benchmark,
with grouped separation and nested configuration selection.

Two calibration policies express different deployment preferences.
$B_{10}$ selects the lowest float64 threshold satisfying an empirical
10\% calibration false positive rate (FPR) budget.
$B_{F1}$ maximizes calibration F1, breaking ties
by lower FPR and then higher threshold. Neither uses the outer test, and
one global threshold applies to both sources. Inner selection prioritizes
mean validation FPR at most 10\%, then equal-source F1, recall, smaller $K$,
and smaller dimension. If none meets the budget, lower FPR takes priority.
Appendix~\ref{app:training} gives exact configurations and tie handling.

Unsafe is the positive class. We report its F1 and recall, and
$\mathrm{FPR}=\mathrm{FP}/(\mathrm{FP}+\mathrm{TN})$, the fraction of safe
calls incorrectly blocked. Confusion counts are pooled across outer folds
separately for each of five seeds; reported metrics average these five
per-seed values. No score ensemble is used. Deterministic baselines have
one prediction per candidate. Area under
the receiver operating characteristic curve (AUROC) is averaged within
folds to avoid mixing fitted score scales. Confidence intervals use 5,000
source-stratified root bootstrap resamples, conditional on the fitted
models; they exclude retraining uncertainty and shared-fold dependence.

\subsection{Detection utility and alternative scorers}
\label{sec:scorer-evaluation}
Table~\ref{tab:quality} compares scorer choices at the numerical interface
on the same 1,252 candidates. Logistic
regression (LR), a radial-basis-function support vector machine (RBF-SVM),
and direct cosine use the same fitted input as SealCall. LR and RBF-SVM
select $C\in\{0.1,1,10\}$ by the same nested criterion. The center baseline
uses the normalized mean reference in the learned space; it is not trained
with a separate center objective.

\begin{table*}[t]
\caption{Detection quality on the same 1,252 candidates. Each policy is
calibrated before outer testing. AUROC averages within folds. Lower FPR
is better; other metrics are higher-is-better.}
\label{tab:quality}
\centering
\begin{tabular*}{\textwidth}{@{\extracolsep{\fill}}lrrrrrr@{}}
\toprule
& \multicolumn{3}{c}{$B_{10}$} & \multicolumn{2}{c}{$B_{F1}$} & \\
\cmidrule(lr){2-4}\cmidrule(lr){5-6}
Scorer & F1 & Recall & FPR & F1 & FPR & AUROC\\
\midrule
SealCall reference set & 0.812 & 0.823 & 0.153 & 0.840 & 0.112 & 0.921\\
Same-space center & 0.805 & 0.805 & 0.145 & 0.818 & 0.120 & 0.917\\
Same-input LR & 0.899 & 0.936 & 0.110 & 0.913 & 0.089 & 0.967\\
Same-input RBF-SVM & 0.889 & 0.948 & 0.137 & 0.905 & 0.052 & 0.970\\
Direct cosine & 0.627 & 0.540 & 0.137 & 0.666 & 0.449 & 0.736\\
\bottomrule
\end{tabular*}
\end{table*}

SealCall achieves F1 of 0.812 at $B_{10}$ and 0.840 at $B_{F1}$, with
conditional 95\% intervals [0.773, 0.849] and [0.804, 0.874]. Its
$B_{10}$ test FPR of 15.26\% illustrates the difference between a
calibration budget and future behavior. At $B_{F1}$, source-specific F1
is 0.745 on Harm and 0.888 on Dojo, giving equal-source F1 of 0.816.
The larger Dojo source therefore contributes disproportionately to pooled
F1. Appendix~\ref{app:source} provides source-resolved comparisons.

The scorer comparisons identify distinct choices at the numerical
interface. Learning the representation improves over direct cosine;
LR and RBF-SVM achieve F1 of 0.913 and 0.905. The reference scorer's
$B_{F1}$ F1 difference from
LR is $-0.072$, with a familywise paired interval of $[-0.112,-0.035]$;
the corresponding difference from RBF-SVM is $-0.065$
$[-0.106,-0.031]$. Reference scoring returns similarities to individual
examples, while LR and RBF-SVM provide stronger discrimination on these
data. Section~\ref{sec:encrypted-evaluation} evaluates the computation and
communication costs of these choices within the same delegation interface.

\paragraph{External guard reference.}
A separately fixed set of 416 shared candidates compares native guard
outputs with SealCall's corresponding outer predictions. StepGuard-4B~\cite{stepguard}
and Qwen3Guard-Gen-0.6B~\cite{qwen3guard} use Q4\_K\_M weights, native prompts and parsers,
and the task, candidate, and strict prior history; tool names are supplied
without fabricating missing schemas. Their F1 values are 0.882 and 0.411,
versus SealCall's 0.837 at $B_{F1}$. There are no invalid responses on
this subset. These are same-case reference results, with unknown training
overlap, quantization, and incomplete tool metadata. Their 416-case scores
are not directly comparable to Table~\ref{tab:quality}.

\subsection{Information retained by audit representations}
\label{sec:representation-evaluation}
Matched ablations remove the text embedding or structured-feature block
while preserving each outer fold's selected configuration, reference
identities, partitions, seed, optimizer, and input width. Removed coordinates
are zeroed before joint normalization, followed by refitting and independent
calibration on the original calibration groups. The 50 additional fits
therefore test the contribution of a block under the selected architecture.

\begin{table}[t]
\caption{Representation comparisons at $B_{10}$ on all 1,252 candidates.
The lower panel uses a shared, separately block-normalized LR readout.}
\label{tab:representation}
\centering
\begin{tabular}{lrrr}
\toprule
Representation / readout & Harm F1 & Dojo F1 & Pooled F1\\
\midrule
SealCall full & 0.728 & 0.854 & 0.812\\
Without text embeddings & 0.599 & 0.857 & 0.769\\
Without structured features & 0.549 & 0.790 & 0.718\\
\midrule
BGE + features / LR & 0.792 & 0.929 & 0.879\\
TF-IDF + features / LR & 0.808 & 0.919 & 0.878\\
Features only / LR & 0.742 & 0.939 & 0.870\\
Tool-name probe / LR & 0.770 & 0.889 & 0.846\\
\bottomrule
\end{tabular}
\end{table}

The full representation attains higher pooled F1 than either ablation
(Table~\ref{tab:representation}). Removing text embeddings reduces
$B_{10}$ F1 by 0.043; the familywise interval is $[-0.003,0.090]$.
At $B_{F1}$, the corresponding reduction is 0.039
$[-0.009,0.092]$. Both intervals include zero after correction across
12 comparisons. Text changes the error trade-off: at $B_{10}$ it increases
recall and FPR, while at $B_{F1}$ it reduces FPR.

The matched LR comparison uses the same compact text, structured features,
fitting groups, sample weights, and nested selection. Each text/feature
block is normalized separately before joint normalization. A fitting-only
TF-IDF representation combines word 1--2 and character 3--5 grams.
Its F1 of 0.878 is close to BGE-LR's 0.879. Features-only and tool-name
controls also perform well, especially on Dojo. These results show that
lexical and structured cues account for much of the discrimination on
this benchmark, with substantial variation across sources.

\paragraph{Changes in expression and authorization.}
Two fixed diagnostics test different forms of generalization without
refitting or recalibration. For 100 anchors balanced over source and label,
two task restatements preserve the candidate and history. We re-extract
the compact context and structured features from each revised task. At $B_{10}$,
F1 is 0.785 on originals, 0.791 on the first restatements, and 0.787 on
the second. The fraction
correct on all three is 0.776, exceeding the no-text-embedding model by
0.070; a descriptive bootstrap over 93 distinct task bodies gives
$[0.010,0.130]$. The corresponding $B_{F1}$ interval crosses zero.
An AI model generated the restatements, and another model reviewed them;
there was no human adjudication. Six ambiguous roots remain in the main diagnostic.

A second construction contains 80 policy worlds, each with one safe and
one unsafe case in two renderings, totaling 320 cases. It changes recipient
binding, purpose, source authority, or prior conditions while keeping the
candidate and history fixed. At either operating point, no world has all
four decisions correct. A parser of the construction grammar attains
1.000 on this measure. Four families share two surface grammars, and role
facts are asserted in text rather than authenticated runtime identity.
This diagnostic isolates sensitivity to changed authorization: the current
representation and scorer do not reliably distinguish these policy changes.
Capturing the altered recipient, purpose, or source authority is therefore
a distinct representation objective from stability under task restatement.

\subsection{Encrypted fidelity and computation costs}
\label{sec:encrypted-evaluation}
The cost representative was fixed before inspecting quality: outer fold
0, seed 20261016, with $d=256$ and $K=16$. Experiments use an Intel Core
Ultra 5 225H, 32\,GB RAM, one thread, and TenSEAL 0.3.18. CKKS parameters
are polynomial modulus degree 8192, coefficient-modulus bit sizes
$[60,40,40,60]$, and scale $2^{40}$. Actual encryption and homomorphic
operations run in one process with distinct secret and public contexts;
the server context has no secret key. These measurements exclude cloud
network latency and do not test process isolation.

\begin{table}[t]
\caption{Quality and private scoring cost for one pre-selected model
per scorer. F1 uses $B_{F1}$ on that model's 251-case outer test fold;
Table~\ref{tab:quality} aggregates all outer folds. Response sizes are decimal MB.}
\label{tab:cost}
\centering
\begin{tabular}{lrrrr}
\toprule
Scorer & F1 & Time (s) & Packed MB & Raw MB\\
\midrule
SealCall & 0.895 & 0.176 & 0.131 & 3.761\\
Center & 0.845 & 0.016 & 0.235 & 0.235\\
LR & 0.909 & 0.043 & 0.235 & 0.235\\
RBF-SVM & 0.915 & 15.267 & 0.131 & 93.794\\
\bottomrule
\end{tabular}
\end{table}

After three warmups, 50 queries are timed for the reference set, center,
and LR, and 20 for the complete 399-support-vector RBF model. LR/RBF inputs
have 1,150 coordinates in this fold. The RBF control returns encrypted
linear terms for trusted kernel evaluation, using the same allowance for
trusted postprocessing as SealCall's maximum. Table~\ref{tab:cost} pairs
each model's quality with its own costs. LR is both more accurate and faster
than the reference scorer on this representative; SealCall's packed
response is smaller. RBF offers higher quality at much greater arithmetic
cost. Confidential delegation consequently requires a joint choice of
representation, scorer, output interface, and communication layout.

Packing reduces the reference response by 96.5\%, from 3.761 to
0.131\,MB, and also benefits RBF without reducing its dot-product count.
The query is 0.331\,MB and the encrypted reference set 5.302\,MB.
Trusted extraction, uncached BGE encoding, and projection average
1.47, 484.95, and 0.55\,ms. Query-matched sums of trusted preparation
and private scoring average 0.663\,s, with empirical 95th percentile
1.158\,s, excluding network and scheduling. Local plaintext scoring after
preparation takes approximately $19\,\mu$s. These measurements quantify
the processing cost of a confidential remote audit interface. The service
model is motivated by access to maintained audit knowledge and operations;
the encrypted scoring time measures the cost of applying a selected
scorer across that boundary while protecting the query.

Across 760 model--layout numerical records, both frozen operating policies
give the same decisions as plaintext. The records comprise two layouts
of 251 reference-set, 53 center, 53 LR, and 23 RBF queries. The maximum
final-score error for the
packed reference scorer is $1.52\times10^{-6}$; across all layouts and
scorers it reaches approximately $1.19\times10^{-5}$. Thus encrypted
evaluation preserves the selected scorers' decisions across these records.

\subsection{Execution binding and evaluation scope}
\label{sec:execution-evaluation}
Input reconstruction validates 3,066 selected source spans and all 1,525
compact contexts. Their total text uses 21.6\% of the reference characters,
quantifying the text reduction achieved by the selection rule. All 36
focused adapter tests pass, covering replay, changed candidates, stale
state, and failure closure. Together these checks validate context
reconstruction and call binding for the tested cases. Historical controlled
OpenClaw probes used an earlier scorer; the current pipeline's detection
results come from the offline evaluation above.

The evaluation connects detection utility, encrypted-decision agreement,
and execution binding within a common audit workflow. The scorer
comparisons expose concrete choices in detection quality, computation,
and communication. Applying this workflow to a deployment requires
calibration on representative tasks and reliable tool mediation under
the trust model in Section~\ref{sec:confidentiality}.

\section{Related Work}
\label{sec:related-work}
\paragraph{Specialized security services.}
Managed security services pool expertise and threat knowledge while
introducing sensitive-data access across organizational
boundaries~\cite{nist-incident-response}. Contextual agent auditing brings
user instructions, call arguments, and interaction history to this
boundary. Our paradigm supports remote audit computation while retaining
these plaintext inputs and execution control in the user control domain.

\paragraph{Agent safety and runtime enforcement.}
AgentDojo, AgentHarm, and InjecAgent expose risks arising from tool use,
malicious requests, and untrusted external
content~\cite{agentdojo,agentharm,injecagent}. ToolSafe's step-level guard
conditions its judgment on the task, candidate action, and interaction
history and supplies feedback before execution~\cite{toolsafe}.
StepGuard similarly studies step-level checks and safety--utility
balancing~\cite{stepguard}. These tasks establish the need for contextual
auditing. Our paradigm addresses its delegation: constructing the audit
representation in a trusted environment, computing on ciphertexts at the
service, and retaining the final decision and release authority at the
controller.

AgentSpec expresses runtime constraints through triggers, predicates,
and enforcement mechanisms~\cite{agentspec}. CaMeL separates a privileged
model that plans from the user query from a quarantined model that parses
untrusted data, then enforces capability-based information-flow
policies~\cite{camel}. These approaches regulate source authority, agent
behavior, and permitted data flows. Our three-domain auditing model additionally
treats the audit provider as an observer of private inputs; its processing
layers connect this boundary to representation, encrypted computation,
and trusted mediation. SealCall uses a learned risk scorer and
call-binding checks.

\paragraph{Private inference and similarity computation.}
CryptoNets demonstrates prediction on encrypted inputs without supplying
the cloud with a decryption key~\cite{cryptonets}. GAZELLE combines HE
with two-party computation for private neural inference and protects model
parameters as well as client inputs under its specified
functionality~\cite{gazelle}. SANNS studies secure approximate nearest
neighbors, including semi-honest security for its search
protocols~\cite{sanns}. Our design question is how to turn private task
context into a useful numerical audit input and connect a confidential
result to execution. The paradigm makes representation, delegated
arithmetic, and the output interface explicit design choices. SealCall
uses CKKS for encrypted similarity computation~\cite{ckks}, returning
individual similarities for trusted aggregation. Section~\ref{sec:delegation}
specifies the information disclosed by this output interface.

\paragraph{Compression and numerical representations.}
Selective Context and LLMLingua reduce language-model input
length~\cite{selective-context,llmlingua}; sentence embedding methods
provide vectors suitable for similarity comparison~\cite{sentence-bert,bge-cpack}.
The trusted representation layer adapts these ideas to the relationships
between task constraints, call arguments, and source roles. SealCall
selects argument-linked historical text, combines it with structured
features, and learns a projection for risk scoring. Its ablations and
authorization diagnostics evaluate what this representation retains.
Embedding inversion results~\cite{embedding-inversion} motivate encrypting
the resulting vectors at the audit boundary. Information preservation
and ciphertext confidentiality consequently receive separate analysis.

\section{Conclusion}
\label{sec:conclusion}
Specialized audit services offer maintained safety capabilities for agents
acting through Web services. Obtaining these services can expose private
tasks, calls, and history through the audit interface. We characterize
this exposure with a three-domain auditing model and a paradigm for trusted
representation, encrypted computation, and trusted decision and execution.
Its interfaces allow delegated risk scoring while plaintext access,
result interpretation, and call release remain in the user control domain.

SealCall combines learned representations with
encrypted reference similarities. Grouped detection, encrypted--plaintext
decision agreement, cost measurements, and execution checks demonstrate
the workflow and quantify its detection, computation, and communication
trade-offs. Future work will develop representations sensitive to changes
in authorization, optimize encrypted computation, and extend
integration with agent runtimes such as Codex and OpenCode.
With trusted preparation, isolated keys, and reliable mediation, external
auditing can protect agents before execution without plaintext user context.

\bibliographystyle{ACM-Reference-Format}
\bibliography{references}
\appendix
\section{Implementation, Fitting, and Calibration Details}
\label{app:training}
Context extraction matches scalar argument literals in strictly earlier
observations, case-insensitively with token boundaries, excluding empty
and boolean values. It selects the first two and latest matches, extends
each by up to 96 characters on either side, and deduplicates shared spans.
Selected text retains source roles and offsets; absent matches remain
explicitly missing. The frozen BGE encoder uses 384-token windows with
stride 352. Window embeddings are normalized, mean-pooled, and normalized
again. Structured features encode tool identity, input roles, argument
paths and types, history presence, and whether literals occur in the task,
history, both, or neither. Missing, unknown, and unseen values have explicit
encodings; inference retains the fitted feature width.

The split seed is 20261018. Inner selection uses seed 20261016; final
fitting uses seeds 20261016--20261020. The frozen BGE model revision is
\texttt{5c38ec7c405ec4b44b94cc5a9bb96e735b38267a}.
The semantic block has 384 coordinates; structured-feature dimensions
depend on the vocabulary fitted in each outer fold. The representative
fold has 1,150 total input coordinates before projection. Cached text
embeddings are frozen independently of the risk labels.

Let $B(q)$ denote unsafe references remaining after excluding the query's
task root. The training risk score and bank objective are
\begin{align}
 r(q)&=\max_{p\in B(q)}z_\theta(q)^\top z_\theta(p),\\
 L_{\mathrm{bank}}&=\mathbb{E}_q\,\operatorname{BCE}
       (\sigma(\alpha(r(q)-t)),y_q),
\end{align}
where $\operatorname{BCE}$ is binary cross-entropy, $\sigma$ is the
logistic function, $t$ is a learned training-only threshold, and $\alpha=20$.
For unsafe queries, an optional hard-example loss contrasts the best
eligible reference with safe in-batch examples $\mathcal{N}$:
\begin{equation}
 L_{\mathrm{hard}}=\mathbb{E}_{q:y_q=1}
 \operatorname{softplus}\!\left(\alpha\left[m+
 \max_{n\in\mathcal{N}}z(q)^\top z(n)-r(q)\right]\right),
\end{equation}
with $m=0.1$. The candidate objective is
$L_{\mathrm{bank}}+\beta L_{\mathrm{hard}}$, $\beta\in\{0,1\}$.
Deployment recalibrates a separate threshold.

AdamW uses learning rate $10^{-3}$, batch size 128, a maximum of 120
epochs, a minimum of 24, and stopping patience 12. The four candidates are
bank-only at 128 dimensions, bank-only at 256, bank plus hard loss at
128, and bank-only at 128 with weight decay 0.01. The first three use
weight decay $10^{-3}$. Each batch samples both classes equally, then
samples roots uniformly within a class and examples within a root.
Early stopping uses standardized partial AUROC up to FPR 0.1. All
configurations use at most 32 source-balanced unsafe reference instances,
with one per root; the selected deployment reference count is a prefix
of the same hash order, selected from $K\in\{1,4,8,16,32\}$.

Three inner folds compare four maps and five prefixes in each outer fold,
requiring 60 map fits because prefixes share a fit. Five final seeds add
25 fits. Inner fits reserve one of five group folds for stopping and
provisional calibration. Final development reserves one of ten group
folds for stopping, while final calibration remains separate. Feasible
configurations have mean inner FPR at most 0.10; selection maximizes
equal-source F1, then recall, and prefers smaller $K$, smaller $d$, and
lexical configuration name for remaining ties. If none is feasible,
lower FPR is the first criterion. There are no test-based split retries.

For $B_{10}$, sort calibration safe scores as
$r_0\geq\cdots\geq r_{n-1}$ and set $j=\lfloor0.1n\rfloor$.
The threshold is the next float64 value above $r_j$, using the comparison
$s\geq\tau$. This rejects at most $j$ calibration safe examples, including
in the presence of tied scores. $B_{F1}$ scans distinct scores and a
threshold above the maximum, then maximizes F1 with the tie rules in
Section~\ref{sec:evaluation}. All reported decisions use explicit float64
threshold comparisons.

\section{Source-Resolved Utility and Diagnostics}
\label{app:source}
Table~\ref{tab:source} separates the two sources at $B_{F1}$ on the full
binary benchmark. These metrics use the same source-independent threshold
for a given fitted model. Equal-source macro F1 is the mean of source
F1 values within each seed, rather than F1 from equal-sized resampling.

\begin{table}[t]
\caption{Source-resolved $B_{F1}$ quality on the 1,252-candidate benchmark.}
\label{tab:source}
\centering
\begin{tabular}{lrrrr}
\toprule
& \multicolumn{2}{c}{Harm} & \multicolumn{2}{c}{Dojo}\\
\cmidrule(lr){2-3}\cmidrule(lr){4-5}
Scorer & F1 & FPR & F1 & FPR\\
\midrule
SealCall & 0.745 & 0.132 & 0.888 & 0.105\\
Center & 0.689 & 0.122 & 0.878 & 0.119\\
LR & 0.869 & 0.146 & 0.938 & 0.066\\
RBF-SVM & 0.841 & 0.102 & 0.941 & 0.031\\
Direct cosine & 0.621 & 0.527 & 0.689 & 0.418\\
\bottomrule
\end{tabular}
\end{table}

The LR lexical comparison fits word 1--2 and character-within-word 3--5
grams, each capped at 12,000 features, with minimum document frequency 2
and sublinear term frequency. Text and standardized structured features
are normalized separately, then jointly. BGE-LR follows the same
block-normalized interface, which differs from the original standardized
same-input LR control. Both use fitting-only preprocessing and nested
selection over $C\in\{0.1,1,10\}$.

The semantic-effect correction treats two contrasts, two operating
policies, and three metrics as one 12-comparison family. The primary
scorer comparison corrects four contrasts within each policy and metric.
Both use 5,000 source-stratified paired root bootstraps conditional on the
fitted models. In the restatement diagnostic, excluding the six ambiguous
roots yields all-three correctness of 0.766 rather than 0.776. At $B_{F1}$,
the improvement over the no-text model is 0.032 with an interval crossing
zero. Policy-world labels are derived from explicit construction facts
checked by a parser; this does not make them human-adjudicated deployment
examples. New worlds are assigned outer models by a fixed hash, and seeds
are evaluated individually before averaging.

Restatement construction first freezes 25 anchors in each source--label
stratum, selecting roots by hash order without consulting predictions.
An AI assistant writes two substantive task restatements while retaining
tool-order hints; an independent model reviews them without model scores.
Structural checks verify unchanged candidate/history fields and valid
citations. These checks do not certify equivalence of unrestricted
language. All selected roots are retained, with review failures examined
in sensitivity analysis. The retained protocol does not pin the generation
and review model versions; the frozen texts and review records are the
reproducible diagnostic inputs.

The artifact includes two levels: replacing text while keeping original
structured features, and re-extracting the whole input after replacing the
task. Main-text results use the latter. The 80 authorization worlds are
constructed independently, with 20 in each relation family and two
renderings per safe/unsafe pair. A predicate evaluator and a separately
implemented parser must agree before scoring. No construction, review
result, or diagnostic score changes the trained maps or thresholds.

The external guard subset contains two previously fixed identity sets of
149 and 267 cases, totaling 416 cases and 61 source roots. Prompts are
matched to the current strict prefixes; frozen native responses are reused
without prompt selection. Training overlap with these public tasks is
unknown. A resource-limited extension with only 12 completed cases is
excluded from quality claims. A separate 140-case constructed diagnostic
also remains separate from the full-benchmark comparison.

\section{Cost Scope and Reproducibility}
\label{app:cost}
The encrypted representative uses CKKS degree 8192 and scale $2^{40}$,
with modulus bits $[60,40,40,60]$ and no bootstrapping. Once-only
reference encryption and serialization take 0.066\,s and produce
5.302\,MB. Packed-query encryption, server computation and packing, and
controller finalization average 4.56, 171.11, and 0.43\,ms, respectively.
The private-stage empirical 95th percentile is 191.61\,ms across 50 timed
queries. Numerical records include three warmups for each model and
layout; timed means exclude these warmups. Fixed layouts reveal dimensions and
reference counts. Network transport, scheduling, process isolation,
concurrent users, and cryptographic key lifecycle operations beyond the
measured setup are not included in these costs.

The retained research artifacts contain source revisions, parsed
candidates, root assignments, compact contexts, encoder caches, fitting
configurations, learned weights, reference identities, thresholds, and
per-case predictions. SHA-256 manifests bind the saved inputs and outputs;
the manuscript's numerical summaries can be reconstructed from these
artifacts. The source data release is the ToolSafe snapshot with revision
\texttt{46358fa}; full identifiers are retained in the artifact manifest.

\section{Action-Stratified Detection Quality}
\label{app:action-strata}
Table~\ref{tab:action-slices} partitions the 1,252 candidates by a frozen
tool-name rule applied equally to safe and unsafe cases. The five strata
partition cases; a task root can appear in several strata. SealCall
recalls 89.6\% of unsafe sends and 92.9\% of unsafe writes at $B_{10}$,
with lower recall on publish and other. The execute proxy contains five
roots. Tool names name the stratum and do not establish actual effects
or authorization.

\begin{table*}[t]
\caption{Detection within frozen tool-name strata at $B_{10}$. Assignment
uses one naming rule for safe and unsafe cases.}
\label{tab:action-slices}
\centering
\begin{tabular*}{\textwidth}{@{\extracolsep{\fill}}lrrrrrrrrr@{}}
\toprule
& & & & \multicolumn{2}{c}{SealCall} & \multicolumn{2}{c}{Same-input LR} & \multicolumn{2}{c}{TF-IDF + features, LR}\\
\cmidrule(lr){5-6}\cmidrule(lr){7-8}\cmidrule(lr){9-10}
Proxy & Roots & Safe & Unsafe & Recall & FPR & Recall & FPR & Recall & FPR\\
\midrule
Send & 97 & 73 & 163 & 0.896 & 0.260 & 0.939 & 0.274 & 0.969 & 0.411\\
Write & 72 & 148 & 187 & 0.929 & 0.284 & 0.995 & 0.122 & 0.989 & 0.203\\
Execute & 5 & 9 & 9 & 0.178 & 0.022 & 0.333 & 0.111 & 0.222 & 0.000\\
Publish & 57 & 36 & 103 & 0.715 & 0.294 & 0.913 & 0.278 & 0.981 & 0.556\\
Other & 147 & 451 & 73 & 0.622 & 0.083 & 0.890 & 0.067 & 0.767 & 0.058\\
\bottomrule
\end{tabular*}
\end{table*}

Table~\ref{tab:fold-selection} records the configuration chosen in each
outer fold and the sizes of the disjoint fit, stopping, calibration, and
test partitions. Selected dimension and hard-example loss vary across
folds. Every deployment bank is a prefix of one source-balanced unsafe
reference order, so $K$ changes the bank size without changing the
training examples behind it.

\begin{table}[H]
\caption{Outer-fold selections and disjoint final roles
(decisions/task roots). $K$ is a prefix of the same unsafe reference order.}
\label{tab:fold-selection}
\centering
\begin{tabular}{rlrrrrr}
\toprule
Fold & Configuration & $K$ & Fit & Stop & Cal. & Test\\
\midrule
1 & bank256 & 16 & 721/115 & 79/12 & 201/32 & 251/43\\
2 & hard128 & 16 & 717/117 & 85/12 & 202/32 & 248/41\\
3 & hard128 & 1 & 719/118 & 79/13 & 203/32 & 251/39\\
4 & bank128 & 1 & 719/116 & 81/13 & 199/34 & 253/39\\
5 & hard128 & 8 & 719/117 & 82/13 & 202/32 & 249/40\\
\bottomrule
\end{tabular}
\end{table}

\section{Expanded Semantic Diagnostics}
\label{app:semantic-diagnostics}
Table~\ref{tab:semantic-diagnostics} places the two frozen diagnostics
side by side. Restatement stability follows each method's original F1:
lexical readouts keep all-three correctness at 0.850 or above, and
removing text embeddings drops it to 0.706. On the constructed worlds,
no tested method gets all four decisions correct in any world, and
pair-level correctness stays at 0.025 or below. These constructions
change policy meaning in text; they do not authenticate runtime identity.

\begin{table*}[t]
\caption{Frozen-model semantic diagnostics at $B_{10}$. Keeping meaning:
100 source-root anchors with two task restatements (300 candidates).
Changing policy relations: 80 constructed worlds with two labels and two
renderings (320 candidates). Higher correctness is better.}
\label{tab:semantic-diagnostics}
\centering
\begin{tabular*}{\textwidth}{@{\extracolsep{\fill}}lrrrrrr@{}}
\toprule
& \multicolumn{3}{c}{Meaning preserved} & \multicolumn{3}{c}{Relation changed}\\
\cmidrule(lr){2-4}\cmidrule(lr){5-7}
Method & Orig. F1 & Reword. F1 & All 3 correct & Pair 0 & Pair 1 & All 4 correct\\
\midrule
SealCall & 0.785 & 0.789 & 0.776 & 0.013 & 0.013 & 0.000\\
Without text embeddings & 0.669 & 0.672 & 0.706 & 0.000 & 0.000 & 0.000\\
Without structured features & 0.675 & 0.677 & 0.708 & 0.005 & 0.013 & 0.000\\
TF-IDF + features / LR & 0.876 & 0.887 & 0.870 & 0.013 & 0.013 & 0.000\\
BGE + features / LR & 0.895 & 0.891 & 0.880 & 0.013 & 0.000 & 0.000\\
Same-input LR & 0.871 & 0.887 & 0.850 & 0.025 & 0.025 & 0.000\\
Tool-name probe / LR & 0.792 & 0.792 & 0.800 & 0.000 & 0.000 & 0.000\\
\bottomrule
\end{tabular*}
\end{table*}

Table~\ref{tab:relation-families} splits the 80 worlds by relation
family. SealCall reaches F1 0.607 on source authority and 0.171 on
purpose scope; no family yields an all-four-correct world. The four
families share two surface grammars, and each world pairs one safe and
one unsafe case under the same candidate and history.

\begin{table}[H]
\caption{Constructed-policy diagnostics by relation family (20 worlds
per family). Each world pairs one safe and one unsafe case.}
\label{tab:relation-families}
\centering
\setlength{\tabcolsep}{2pt}
\begin{tabular}{llrr}
\toprule
Method & Family & F1 & All-4\\
\midrule
SealCall & recipient\_binding & 0.455 & 0.000\\
SealCall & purpose\_scope & 0.171 & 0.000\\
SealCall & source\_authority & 0.607 & 0.000\\
SealCall & prior\_precondition & 0.598 & 0.000\\
Without text embeddings & recipient\_binding & 0.421 & 0.000\\
Without text embeddings & purpose\_scope & 0.323 & 0.000\\
Without text embeddings & source\_authority & 0.561 & 0.000\\
Without text embeddings & prior\_precondition & 0.590 & 0.000\\
TF-IDF + features / LR & recipient\_binding & 0.629 & 0.000\\
TF-IDF + features / LR & purpose\_scope & 0.196 & 0.000\\
TF-IDF + features / LR & source\_authority & 0.637 & 0.000\\
TF-IDF + features / LR & prior\_precondition & 0.667 & 0.000\\
Same-input LR & recipient\_binding & 0.506 & 0.000\\
Same-input LR & purpose\_scope & 0.091 & 0.000\\
Same-input LR & source\_authority & 0.635 & 0.000\\
Same-input LR & prior\_precondition & 0.655 & 0.000\\
\bottomrule
\end{tabular}
\end{table}

\section{Source Operating Curves}
\label{app:operating-curves}
Figure~\ref{fig:operating} averages held-out ROC curves within each
source, outer fold, and seed before averaging folds and seeds. Scores
from different fitted models are never pooled. The curves describe
ranking across thresholds; they do not select or change the frozen
calibration thresholds.

\begin{figure*}[t]
\centering
\includegraphics[width=\textwidth]{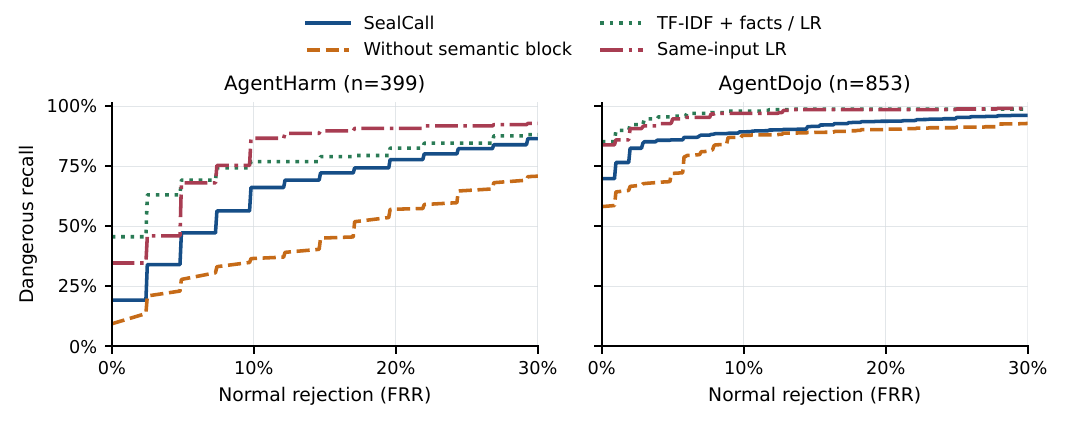}
\caption{Source-specific operating curves from held-out fold scores.
Scores from different fitted models are never pooled.}
\Description{Two panels compare SealCall, removal of text embeddings,
TF-IDF with features, and same-input logistic regression on AgentHarm
and AgentDojo. Each panel averages within-fold ROC curves rather than
pooling score scales.}
\label{fig:operating}
\end{figure*}

\end{document}